\documentclass[12pt,a4paper]{article}

\usepackage[utf8]{inputenc}
\usepackage[T1]{fontenc}
\usepackage{amsmath}
\usepackage{amssymb}
\usepackage{amsthm}
\usepackage{mathtools}
\usepackage{bm}
\usepackage{xcolor}
\usepackage{mathrsfs}
\usepackage{hyperref}

\usepackage[margin=1in]{geometry}
\usepackage{graphicx}
\usepackage{booktabs}
\usepackage{caption}
\usepackage{subcaption}
\usepackage{microtype}

\usepackage{algorithm}
\usepackage{algpseudocode}

\usepackage[
    backend=biber,
    style=numeric,
    sorting=nyt,
    giveninits=true,
    maxbibnames=99,
    doi=false,
    isbn=false,
    url=false,
    eprint=false
]{biblatex}
\DeclareDelimFormat[bib]{finalnamedelim}{%
  \addspace\bibstring{and}\space}
\DeclareFieldFormat[book]{title}{#1}
\DeclareBibliographyDriver{book}{%
  \usebibmacro{bibindex}%
  \usebibmacro{begentry}%
  \printnames{author}%
  \addcomma\space
  \printfield{title}%
  \addcomma\space
  \printlist{publisher}%
  \addcomma\space
  \printlist{location}%
  \addcomma\space
  \printfield{year}%
  \finentry
}
\DeclareFieldFormat[inproceedings]{title}{#1}
\DeclareFieldFormat[inproceedings]{booktitle}{#1}
\DeclareBibliographyDriver{inproceedings}{%
  \usebibmacro{bibindex}%
  \usebibmacro{begentry}%
  \printnames{author}%
  \addcomma\space
  \printfield{title}%
  \addcomma\space
  \printfield{booktitle}%
  \addcomma\space
  \printfield{year}%
  \finentry
}

\DeclareFieldFormat[article]{journaltitle}{#1}
\DeclareFieldFormat[article]{volume}{#1}
\DeclareFieldFormat[article]{title}{#1}
\DeclareFieldFormat[article]{pages}{#1}

\AtEveryBibitem{%
  \clearfield{number}%
  \clearfield{month}%
  \clearfield{day}%
}

\renewbibmacro{in:}{}

\DeclareBibliographyDriver{article}{%
  \usebibmacro{bibindex}%
  \usebibmacro{begentry}%
  \printnames{author}%
  \addcomma\space
  \printfield{title}%
  \addcomma\space
  \printfield{journaltitle}%
  \addcomma\space
  \printfield{volume}%
  \space
  \mkbibparens{\printfield{year}}%
  \addcomma\space
  \printfield{pages}%
  \finentry
}

\usepackage{fontawesome5}

\usepackage{hyperref}

\newcommand{\mailicon}[1]{%
  \href{mailto:#1}{\textsuperscript{\faEnvelope}}%
}

\newtheorem{theorem}{Theorem}[section]
\newtheorem{lemma}[theorem]{Lemma}

\newtheorem{corollary}[theorem]{Corollary}
\newtheorem{proposition}[theorem]{Proposition}

\newtheorem{example}{Example}

\begin{document}

\title{On the structure of constacyclic codes over finite chain rings\footnotetext{2020 Mathematics Subject Classication. Primary: 11T71,94B15,94B65; Secondary: 94B05.\\
Keywords: Finite chain rings, constacyclic codes, MHDR codes, MDS codes.\\
*Corresponding author: Ridhima Thakral.}}
\author{
Vaishali Singh$^{\mailicon{vaishali.phd24maths@pec.edu.in}1}$,
Sucheta Dutt$^{\mailicon{sucheta@pec.edu.in}2}$ and
Ridhima Thakral$^{\mailicon{ridhimathakral@pec.edu.in}*,3}$\\[6pt]
$^{1,2,3}$Department of Mathematics,\\
Punjab Engineering College (Deemed to be University),\\
Sector 12, Chandigarh 160012, India
}
\date{}
\maketitle

\begin{abstract}
In the present paper, we provide an explicit construction for generators of a $\lambda$-constacyclic code $\mathcal{C}$ of arbitrary length $\ell$ over a finite chain ring(FCR) $\mathcal{R}$ in terms of certain minimum degree polynomials of the ring $\mathcal{R}[x]/ \langle x^{\ell}-\lambda \rangle$. Moreover, the proposed construction achieves the minimum possible number of generators. We prove certain properties of this set of generators, using which we obtain a minimal spanning set of $\mathcal{C}$. We also obtain that the rank of $\mathcal{C}$ is $\ell-n_0$, where $n_0$ is the degree of the minimal degree polynomial in $\mathcal{C}$. Finally, we derive necessary and sufficient conditions under which an arbitrary length $\lambda$-constacyclic code $\mathcal{C}$ over $\mathcal{R}$ is Maximum Hamming Distance with respect to Rank(MHDR) as well as Maximum Distance Separable(MDS) in terms of a torsion code of $\mathcal{C}$ over the residue field $\mathbb{F}_q$ of $\mathcal{R}$. We further determine the exact values for $n_0$ for which $\mathcal{C}$ over $\mathcal{R}$ is MHDR.
\end{abstract}

\section{Introduction}
Reliable transmission of information across noisy channels has long demanded mathematical tools and error-correcting codes occupy a central place among them.  Within this domain, the class of cyclic and constacyclic codes received considerable attention due to their close connection between the codewords and ideals of the polynomial ring.

Early studies on cyclic and negacyclic codes over finite fields were established in the 1960s by Berlekamp~\cite{mcdonald1974finite}. Subsequently, in 1991, Castagnoli et al. \cite{castagnoli1991repeated} and Lint \cite{van1991repeated} investigated repeated root cyclic codes. In 1994, it was observed by Hammons et al. \cite{calderbank1994z4} that certain nonlinear binary codes derived from linear codes over $\mathbb{Z}_4$ using Gray maps.

An advancement in coding theory was achieved with the introduction of finite chain rings(FCR) as an underlying algebraic structure. In this direction, the theory of linear and cyclic codes over FCR was developed by Norton and S\u{a}l\u{a}gean \parencite{norton2000structure,norton2002hamming} in 2000. They also found the Hamming distances of such codes. In 2004, Dinh and L\'{o}pez-Permouth \cite{dinh2004cyclic} studied the structural properties of cyclic and negacyclic codes over a FCR, particularly for those lengths that do not divide the characteristic of the residue field of the FCR. Further studies in 2005 by Dinh \cite{dinh2005negacyclic} focused on negacyclic codes defined over Galois rings with length $2^s$. In 2006, it was established by S\u{a}l\u{a}gean \cite{sualuagean2006repeated} that repeated root cyclic and negacyclic codes defined over FCR are generally not principally generated.  In 2008, Dinh studied the structural properties of prime power length negacyclic codes over finite fields~\cite{dinh2008linear}. In 2009, Dinh \cite{dinh2009constacyclic} investigated length $2^s$ constacyclic codes over Galois extension rings of $F_2 + uF_2$. Continuing this line of research, Dinh \cite{dinh2010constacyclic} in 2010, analyzed the structural characterisation and Hamming distances of $(\alpha + u\beta)$-constacyclic codes of length $p^s$ over the ring $F_{p^m} + uF_{p^m}$. In 2012, Dinh\cite{dinh2012repeated} further explored repeated-root constacyclic codes by examining the length $2p^s$ over finite fields. During the same year, Chen et al. \cite{chen2012constacyclic} provided an explicit characterisation of generator polynomials associated with constacyclic codes of length $l^tp^s$ over a finite field, where $p$ denotes the field characteristic and $l$ is a prime different from $p$. In 2013, Cao \cite{cao2013constacyclic} examined arbitrary length $(1 + w\gamma)$-constacyclic codes over FCR. Also in that year, Dinh \cite{dinh2013structure} investigated length $3p^s$ repeated-root constacyclic codes over finite fields, focusing on their structural and dual properties. Further developments were made in 2014 by Chen et al. \cite{chen2014repeated}, who characterised repeated-root constacyclic codes of length $\ell p^s$ over finite fields. In 2015, Raka \cite{raka2015class} determined all $\mu$-constacyclic codes of length $\ell^n$. She also characterised repeated-root $\lambda$-constacyclic codes of length $\ell^n p^s$ over finite fields. In 2016, Chen et al. \cite{chen2016constacyclic} established the structure of $\lambda$-constacyclic codes of length $2p^s$ over the ring $F_{p^m} + uF_{p^m}$, where $\lambda = \alpha + u\beta$ is a non-square unit and $\alpha, \beta \in F_{p^m}$ are non-zero elements. In 2017, Dinh et al. \cite{dinh2017repeated} analysed prime power length repeated-root constacyclic codes over finite commutative chain rings. In 2018, Cao et al. \cite{cao2018constacyclic} described the structure of $\lambda$-constacyclic codes of length $np^s$ over the ring $F_{p^m}[u]/\langle u^2 \rangle$, where the length of the code and characteristic of the residue field are coprime. In 2019, Sharma and Sidana \cite{sharma2018structure} examined the structural characteristic and distance distribution of repeated-root constacyclic codes with prime power length over finite commutative chain rings.

Monika et al. \cite{monika2021cyclic} derived the structure of cyclic codes of arbitrary length over a FCR in 2021. Later in 2024, Dalal et al. \cite{dalal2024mds} established a unique generating set for cyclic codes of arbitrary length over a FCR, along with necessary and sufficient conditions for these codes to be Maximum Hamming Distance with respect to Rank(MHDR) and Maximum Distance Separable(MDS). Extending her work, the present paper investigates constacyclic codes of arbitrary length over FCR and obtains a minimal set of generators for such codes. We also obtain necessary and sufficient conditions for these constacyclic codes to be MHDR as well as MDS.
 
The rest of this paper is arranged as follows. The basic algebraic background and notation are introduced in Section~\ref{sec:prelim}. We provide a generating set for $\mathcal{C}$ in terms of certain minimum degree polynomials in the ring $\mathcal{R}[x]/ \langle x^{\ell}-\lambda \rangle$, where $\mathcal{C}$ is a $\lambda$-constacyclic code of arbitrary length $\ell$ over a FCR $\mathcal{R}$ in section~\ref{sec:structure}. In the same section, we prove certain properties of this set of generators, using which we obtain a minimal spanning set of $\mathcal{C}$. We obtain the rank of $\mathcal{C}$. Finally, we obtain necessary and sufficient conditions for an arbitrary length $\lambda$-constacyclic code $\mathcal{C}$ over $\mathcal{R}$ to be MHDR as well as to be MDS in terms of a torsion code of $\mathcal{C}$ over the residue field $\mathbb{F}_q$ of $\mathcal{R}$. We further determine values of $n_0$ for which an arbitrary length $\lambda$-constacyclic code is MHDR. Some examples are also included to illustrate the results. Finally, the paper is concluded in section~\ref{sec:conc}.

\section{Preliminaries}\label{sec:prelim}
Consider a finite commutative chain ring $\mathcal{R}$. Let the ideal generated by the element $a$ is denoted by $\langle a \rangle$. Assume that $\gamma \in \mathcal{R}$ and the ideal generated by $\gamma$ is the unique maximal ideal in $\mathcal{R}$. Thus, the collection of ideals of $\mathcal{R}$ constitutes a chain under the inclusion ordering, i.e., $\langle 0 \rangle = \langle \gamma^{\rho} \rangle \subset \langle \gamma^{\rho-1} \rangle \subset \cdots \subset \langle \gamma \rangle \subset \langle \gamma^{0} \rangle = \mathcal{R}$, where $\rho$ denotes the nilpotency index of $\gamma$. The factor ring $\mathcal{\overline{R}}=\mathcal{R}/\langle \gamma\rangle$ forms a finite field, say $\mathbb{F}_{q}$, the residue field of the ring $\mathcal{R}$. The ring homomorphism $\mathcal{R} \to \mathbb{F}_{q}$, defined by $a \mapsto \overline{a}$, extends naturally coefficientwise from $\mathcal{R}[x]$ onto $\mathbb{F}_{q}[x]$.
The proposition below holds for a FCR.

\begin{proposition}\label{prop:first}
\cite{dinh2017repeated}
Consider a finite commutative chain ring $\mathcal{R}$ having a maximal ideal $\langle \gamma \rangle$, and let the nilpotency index of $\gamma$ is $\rho$. Then:

\begin{enumerate}
\item[(a)] There exist positive integers $m \ge n$ and a prime $p$ such that $|\mathcal{R}| = p^m$, the residue field $\mathbb{F}_q = \mathcal{R}/\langle \gamma \rangle$ has size $|\mathbb{F}_q| = q= p^n$, and  $\mathcal{R}$ has characteristic equal to a power of $p$.

\item[(b)] There is a unit $\theta \in \mathcal{R}$ whose multiplicative order equals $q-1$. With the Teichm\"uller set $\mathbb{T}=\{0,1,\theta,\theta^2,\ldots,\theta^{q-2}\}$, each element $a \in \mathcal{R}$ has a unique $\gamma$-adic expansion of the form $a=a_0+\gamma a_1+\cdots+\gamma^{\rho-1}a_{\rho-1}$, where $a_i \in \mathbb{T}$ for every $0 \leq i \leq \rho-1$.

\item[(c)] For each integer $i$ with $0 \le i \le \rho$, we have $|\langle \gamma^i \rangle| = q^{\rho - i}$; in particular $|\mathcal{R}| =|\mathbb{F}_q|^{\rho} =q^{\rho}$, hence $m = n\rho$.
\end{enumerate}
\end{proposition}

From the proposition above, it clearly follows that every polynomial $f(x)=a_0+a_1x+\dots+a_nx^n$ with coefficients in $\mathcal{R}$ admits the representation $f(x)=c_0(x)+\gamma c_1(x)+\dots+\gamma^{\rho-1}c_{\rho-1}(x)$, where $c_i(x) \in \mathbb{T}[x]$ for each $i$ satisfying $0 \le i \le \rho-1$.

A code $\mathcal{C}$ of length $\ell$ over $\mathcal{R}$ is any nonempty subset of $\mathcal{R}^{\ell}$. Such a code $\mathcal{C}$ is called linear when it forms an $\mathcal{R}$-submodule of $\mathcal{R}^{\ell}$. Let $\lambda$ be a unit in $\mathcal{R}$. A linear code $\mathcal{C}$ is termed
$\lambda$-constacyclic if it remains invariant under the
$\lambda$-constacyclic shift of each of its codewords., that is, whenever $(\eta_0, \eta_1, \dots, \eta_{\ell-1}) \in \mathcal{C}$, we also have $(\lambda \eta_{\ell-1}, \eta_0, \eta_1, \dots, \eta_{\ell-2}) \in \mathcal{C}$. The cases $\lambda=1~\text{and}~\lambda=-1$ correspond to well-known cyclic and negacyclic codes. Every codeword $(\eta_0, \eta_1, \dots, \eta_{\ell-1}) \in \mathcal{C}$ can be represented as a polynomial $\eta_0+\eta_1x+\dots+\eta_{\ell-1}x^{\ell-1}$ in $\mathcal{R}[x]/\langle x^{\ell} - \lambda \rangle$. Define a map $\pi:\mathcal{R} \to \mathcal{R}[x]/\langle x^{\ell} - \lambda \rangle$ by $\pi(\eta_0, \eta_1, \dots, \eta_{\ell-1})=\eta_0+\eta_1x+\dots+\eta_{\ell-1}x^{\ell-1}~\pmod{(x^{\ell}-\lambda)}$. It is easy to observe that $\mathcal{C}$ is a $\lambda$-constacyclic code if and only if $\pi(\mathcal{C})$ forms an ideal of the ring $\mathcal{R}[x]/\langle x^{\ell} - \lambda \rangle$.

Consider two codewords $\eta = (\eta_0, \eta_1, \dots, \eta_{\ell-1})$ and $\eta' = (\eta'_0, \eta'_1, \dots, \eta'_{\ell-1})$ in $\mathcal{C}$. The Hamming distance between these two codewords is defined as $d(\eta,\eta') = \sum_{i=0}^{n-1} \delta(\eta_i,\eta'_i)$, where $\delta(\eta_i,\eta'_i)=\begin{cases} 0 & \eta_i = \eta'_i\\ 1 & \eta_i \ne \eta'_i \end{cases}$. For a code $\mathcal{C}$, its minimum Hamming distance is given by $d(\mathcal{C})=\min_{{\eta,\eta'} \in \mathcal{C},\, \eta \ne {\eta'}} d({\eta,\eta'})$. The Hamming weight of a codeword $\eta \in \mathcal{C}$, denoted by $wt(\eta)$, is equal to the number of nonzero components of $\eta$. The minimum Hamming weight of the code $\mathcal{C}$ is defined as $wt(\mathcal{C}) = \min_{\eta \in \mathcal{C},\, \eta \ne 0} wt(\eta)$. Observe that for a linear code $\mathcal{C}$, $d(\mathcal{C}) = wt(\mathcal{C})$.

A spanning set of a code $\mathcal{C}$ is a subset $S$ of $\mathcal{C}$ such that every codeword in $\mathcal{C}$ can be expressed as a linear combination of the codewords in $S$ with coefficients in $\mathcal{R}$. The set $S$ is called a minimal spanning set of $\mathcal{C}$ if every proper subset of $S$ fails to span $\mathcal{C}$. Note that a minimal spanning set of $\mathcal{C}$ is not unique. However, it is easy to observe that the number of elements remains unchanged in any two minimal spanning sets of $\mathcal{C}$. The rank of $\mathcal{C}$, denoted by $Rank(\mathcal{C})$, is given by the cardinality of a minimal spanning set of the code $\mathcal{C}$. For any linear code $\mathcal{C}$, $d(\mathcal{C}) \le \ell - Rank(\mathcal{C}) + 1$. The code $\mathcal{C}$ is said to be a Maximum Distance Separable (MDS) code with respect to the Hamming metric if $|\mathcal{C}| = |\mathcal{R}|^{(\ell-d(\mathcal{C})+1)}$. The code $\mathcal{C}$ is termed as a Maximum Hamming Distance with respect to Rank (MHDR) if the equality $d(\mathcal{C}) = \ell - Rank(\mathcal{C}) + 1$ is satisfied.

\section{Constacyclic codes over finite chain rings}\label{sec:structure}
Let $\mathcal{R}$ denotes a finite commutative chain ring whose maximal ideal is generated by $\gamma$, having nilpotency index $\rho$, and residue field $\mathbb{F}_q$. This section aims to describe the generators of $\lambda$-constacyclic codes of length $\ell$ over $\mathcal{R}$. We also derive a minimal spanning set and determine the rank of such codes over $\mathcal{R}$. We further obtain necessary and sufficient conditions under which a constacyclic code is MHDR and MDS.

\subsection{Structure of constacyclic codes over finite chain rings}\label{subsec:one}
Let $\mathcal{C}$ be a $\lambda$-constacyclic code of length $\ell$ over $\mathcal{R}$. Viewing $\mathcal{C}$ as an ideal of the ring $\mathcal{R}[x]/ \langle x^{\ell}- \lambda \rangle$, consider all minimum degree polynomials of $\mathcal{C}$. From these polynomials, choose a polynomial $f_0(x)$ whose leading coefficient contains the lowest possible power of $\gamma$, say $r_0$. Consider the ideal generated by $f_0(x)$ in $\mathcal{C}$. If $\mathcal{C}= \langle f_0(x) \rangle$, then we have obtained the structure of $\mathcal{C}$ in terms of its generator. If $\langle f_0(x) \rangle$ is a proper subset of $\mathcal{C}$, then consider all least degree polynomials in $\mathcal{C} \setminus \langle f_0(x) \rangle$. Among these polynomials, choose a polynomial $f_1(x)$ whose leading coefficient contains the least power of $\gamma$, say $r_1$. Let $\deg(f_0(x))=n_0$ and $\deg(f_1(x))=n_1$. It is easy to see that $n_1>n_0$. We claim that $r_0>r_1$. Otherwise, if $r_1 \geq r_0$, then $f_1(x)-\gamma^{r_1-r_0}x^{n_1-n_0}f_0(x)=g(x), ~\text{where}~ g(x)=0 ~\text{or}~\deg(g(x))<\deg(f_1(x))$. In both cases, $f_1(x)\in \langle f_0(x) \rangle$, which is a contradiction. Thus $r_0 > r_1$. If $\mathcal{C}= \langle f_0(x),f_1(x) \rangle$, we stop the process. Otherwise, consider all the least degree polynomials in $\mathcal{C} \setminus \langle f_0(x),f_1(x) \rangle$. Among these polynomials, choose a polynomial $f_2(x)$ whose leading coefficient contains the least power of $\gamma$, say $r_2$. Let $\deg(f_2(x))=n_2$. It is clear that $n_2>n_1$. Also, by the similar argument as given for proving $r_0>r_1$, we can prove that $r_1>r_2$. Again if $\mathcal{C}= \langle f_0(x),f_1(x),f_2(x)\rangle$, then we are done. Otherwise, we will proceed in the same manner as above to find the polynomials $f_3(x),f_4(x), \dots$ such that $n_0<n_1<n_2< \dots ~\text{and}~ r_0>r_1>r_2> \dots$. Since the possible powers of $\gamma$ lie between $0$ to $\rho-1$, the above process must terminate, i.e., for a positive integer $s$, there exists a least degree polynomial $f_s(x)$ in $\mathcal{C} \setminus \langle f_0(x),f_1(x), \dots, f_{s-1}(x) \rangle$ whose leading coefficient contains the least power of $\gamma$, say $r_s$, such that  $\mathcal{C}= \langle f_0(x),f_1(x),\dots f_s(x) \rangle$. Let $\deg(f_s(x))=n_s$. It is clear that $n_s>\dots >n_1>n_0$ and $r_0>r_1> \dots >r_s$.

The following theorem is obtained from the construction given above.
\begin{theorem}\label{thm:ccspanning}
    Let $\mathcal{C}$ be a $\lambda$-constacyclic code of length $\ell$ over the FCR $\mathcal{R}$ and $f_i(x), 0 \le i \le s$ be the polynomials as described above. Then the set $\{ f_0(x),f_1(x), \dots , f_s(x) \}$ generates  $\mathcal{C}$ as an ideal of the ring $\mathcal{R}[x]/ \langle x^{\ell}- \lambda \rangle$.
\end{theorem}

Note that according to the construction given above, we may end up choosing more than one set of generators of $\mathcal{C}$. However, it is easy to see  the number of generators in any two generating sets of $\mathcal{C}$ chosen by this procedure will be the same. Further, if $\{ f_0(x),f_1(x), \dots, f_s(x)\}~~\text{and}~~\{g_0(x),g_1(x), \dots, g_s(x)\}$ are two generating sets of $\mathcal{C}$ obtained by the above procedure we must have $\deg(f_i(x))=\deg(g_i(x))$ for all $0 \le i \le s$.

\begin{lemma}\label{lem:first}
    Let $\mathcal{C}=\langle f_0(x),f_1(x),\dots, f_s(x)\rangle$ be a $\lambda$-constacyclic code of length $\ell$ over $\mathcal{R}$, where $f_i(x),~0 \le i \le s$, be the polynomials as described above. Then, for every polynomial $q(x) \in \mathcal{C}$ such that $\deg(q(x))<n_i$, $q(x) \in \langle f_0(x),f_1(x), \dots, f_{i-1}(x) \rangle ~\text{for}~ 1 \le i \le s$.
\end{lemma}
\begin{proof}
    Since $\deg(q(x))<n_i$ and $n_i$ is the minimum degree among all the polynomials in $\mathcal{C} \setminus \langle f_0(x),f_1(x),\dots,f_{i-1}(x) \rangle$, it implies that $q(x) \notin \mathcal{C} \setminus \langle f_0(x),f_1(x),\dots,f_{i-1}(x) \rangle$. This proves that $q(x) \in \langle f_0(x),f_1(x),\dots,f_{i-1}(x) \rangle$.
\end{proof}

\begin{theorem}
    Let $\mathcal{C}$ be a $\lambda$-constacyclic code generated by the set $\{ f_0(x),f_1(x), \dots , f_s(x) \}$, where $f_i(x),~ 0 \le i \le s$ are the polynomials as defined above. Then
    \begin{enumerate}
        \item $f_0(x)$, the minimum degree polynomial in $\mathcal{C}$, is uniquely determined.
        \item Each of the $f_i(x)$ is uniquely determined modulo $\langle f_0(x),f_1(x), \dots , f_{i-1}(x) \rangle$ for $1 \le i \le s$.
        \item $\gamma^{r_{i-1}-r_i}f_i(x) \in \langle f_0(x), f_1(x), \dots, f_{i-1}(x) \rangle$ for $1 \le i \le s$. \label{thm:parttwo}
        \item For $0 \le i \le s,~~ f_i(x)=\gamma^{r_i}h_i(x)$, where $h_i(x)$ is a monic polynomial.\label{thm:partthree}
        \item $h_0(x) \mid h_1(x) \pmod{\gamma^{\rho-r_0}}$, $h_{i-1}(x) \mid h_i(x) \pmod{\gamma^{r_{i-2}-r_{i-1}}}$ for $2 \le i \le s$, and $h_s(x) \mid (x^{\ell}-\lambda) \pmod{\gamma^{r_{s-1}-r_s}}$.
    \end{enumerate}
\end{theorem}
\begin{proof}
    \begin{enumerate}
        \item Suppose there exists another polynomial $g_0(x)$ such that $\deg(g_0(x))=n_0$ and the power of $\gamma$ in the leading coefficient of $g_0(x)$ is $r_0$. Then there exists a unit $u_0$ such that $f_0(x)-u_0g_0(x)$ is either $0$ or a polynomial in $\mathcal{C}$ of degree less than $n_0$. This is a contradiction. Therefore, $f_0(x)$ is uniquely determined.
        \item If there exists some $f_i(x),~1 \le i \le s$, that is not uniquely determined modulo $\langle f_0(x),f_1(x), \dots , f_{i-1}(x) \rangle$, then there exists $g_i(x)$ such that $\deg(g_i(x))=n_i$ and the power of $\gamma$ in the leading coefficient of $g_i(x)$ is $r_i$. Then there exists a unit $u_i$ such that $f_i(x)-u_ig_i(x)$ is either $0$ or a polynomial in $\mathcal{C}$ of degree less than $n_i$. By Lemma~\ref{lem:first}, $f_i(x)-u_ig_i(x) \in \langle f_0(x),f_1(x), \dots, f_{i-1}(x) \rangle$. It follows that $f_i(x)\equiv u_ig_i(x)$ modulo $\langle f_0(x),f_1(x), \dots, f_{i-1}(x) \rangle$. Hence, $f_i(x)$ is uniquely determined modulo $\langle f_0(x),f_1(x), \dots, f_{i-1}(x) \rangle$.
        \item Clearly, the polynomial $\gamma^{r_{i-1}-r_i}f_i(x)-x^{n_i-n_{i-1}}u_if_{i-1}(x) \in \mathcal{C}$, has degree less than $n_i$, where $u_i$ is a unit in $\mathcal{R}$. Therefore, by Lemma~\ref{lem:first}, $\gamma^{r_{i-1}-r_i}f_i(x)-x^{n_i-n_{i-1}}u_if_{i-1}(x) \in \langle f_0(x),f_1(x), \dots, f_{i-1}(x) \rangle$. It follows that $\gamma^{r_{i-1}-r_i}f_i(x) \in \langle f_0(x), f_1(x), \dots, f_{i-1}(x) \rangle$ for $1 \le i \le s$.
        \item The proof is carried out by induction on $i$. For $i=0$, let $f_0(x)= \gamma^{r_0}u_0x^{n_0}+ a_{n_0-1}x^{n_0-1}+ \dots+a_1x+a_0$, where $u_0$ is a unit in $\mathcal{R}$ and $a_i \in \mathcal{R},~ 0 \le i \le n_0-1$. If $a_k \not\equiv 0 \pmod{\gamma^{r_0}}$ for some $k, ~0 \le k \le n_0-1$, then $\gamma^{\rho-r_0}f_0(x)$ is a polynomial in $\mathcal{C}$ with degree less than $n_0$, which contradicts the fact that the degree of $f_0(x)$ in $\mathcal{C}$ is minimal. Therefore, $a_k \equiv 0 \pmod{\gamma^{r_0}}~ \forall~ 0 \le k \le n_0-1$ implying that $f_0(x)=\gamma^{r_0}h_0(x)$, where $h_0(x)$ is a monic polynomial. Thus, the result is true for $i=0$.

        Let us assume that the result is true for all $i \le k-1$ i.e. $f_i(x)=\gamma^{r_i}h_i(x)$, where $h_i(x)$ is a monic polynomial for all $0 \le i \le k-1$.

        Now we prove the result for $i=k$. We have $\gamma^{r_{i-1}-r_i}f_i(x) \in \langle f_0(x), f_1(x), \dots, f_{i-1}(x) \rangle$ by part~\ref{thm:parttwo}. Therefore, there exist polynomials $c_0(x),c_1(x),\dots,c_{k-1}(x)$ in $ \mathcal{R}[x]/ \langle x^{\ell}-\lambda \rangle$ such that
        \begin{align*}
            \gamma^{r_{k-1}-r_k}f_k(x) = \sum_{i=0}^{k-1}c_i(x)f_i(x)
            =\sum_{i=0}^{k-1}\gamma^{r_i}c_i(x)h_i(x)=\gamma^{r_{k-1}}\sum_{i=0}^{k-1}\gamma^{r_i-r_{k-1}}c_k(x)h_k(x)
        \end{align*}
        
        It follows that $\gamma^{\rho-r_k}f_k(x)=0$ and therefore, $f_k(x)=\gamma^{r_k}h_k(x)$, where $h_k(x)$ is a monic polynomial.
        
        Hence, $f_i(x)=\gamma^{r_i}h_i(x)$, where $h_i(x)$ is a monic polynomial for $0 \le i \le s$, by mathematical induction.
        
        \item From part~\ref{thm:parttwo}, we have $\gamma^{r_0-r_1}f_1(x) \in \langle f_0(x) \rangle$. Therefore, $\gamma^{r_0-r_1}f_1(x)=c_0(x)f_0(x)$, where $c_0(x) \in \mathcal{R}[x]/\langle x^{\ell}-\lambda \rangle$. This together with part~\ref{thm:partthree} implies that $\gamma^{r_0}(h_1(x)-c_0(x)h_0(x))=0$, which further implies that $h_1(x)-c_0(x)h_0(x)=0 \pmod{\gamma^{\rho-r_0}}$. Therefore, $h_0(x) \mid h_1(x) \pmod{\gamma^{\rho-r_0}}$.

        Again using part~\ref{thm:parttwo}, there exists $c_0(x),c_1(x), \dots, c_{i-1}(x) \in \mathcal{R}[x]/\langle x^{\ell}-\lambda\rangle$ such that
        \begin{align*}
            \gamma^{r_{i-1}-r_i}f_i(x)&=\sum_{k=0}^{i-1}c_k(x)f_k(x)
        \end{align*}
        Above equation together with part~\ref{thm:partthree} implies that
        \begin{align*}
            \gamma^{r_{i-1}}h_i(x)=\sum_{k=0}^{i-1}c_k(x)\gamma^{r_k}h_k(x)=\gamma^{r_{i-1}}c_{i-1}(x)h_{i-1}(x)+\sum_{k=0}^{i-2}c_k(x)\gamma^{r_k}h_k(x)
        \end{align*}
        It follows that
        \begin{align*}
            \gamma^{r_{i-1}}(h_i(x)-c_{i-1}(x)h_{i-1}(x))=\sum_{k=0}^{i-2}c_k(x)\gamma^{r_k}h_k(x)=\gamma^{r_{i-2}}\sum_{k=0}^{i-2}c_k(x)\gamma^{r_k-r_{i-2}}h_k(x)
        \end{align*}
           It follows that $\gamma^{\rho-(r_{i-2}-r_{i-1})}(h_i(x)-c_{i-1}(x)h_{i-1}(x))=0$, which implies that $h_i(x)-c_{i-1}(x)h_{i-1}(x)\equiv 0 \pmod{\gamma^{r_{i-2}-r_{i-1}}}$. Therefore, $h_{i-1}(x) \mid h_i(x) \pmod{\gamma^{r_{i-2}-r_{i-1}}}$ for $2 \le i \le s$.

        Now to prove that $h_s(x) \mid (x^{\ell}-\lambda) \pmod{\gamma^{r_{s-1}-r_s}}$, we see that $\gamma^{r_s}(x^{\ell}-\lambda) \in \mathcal{C}$ which implies that $\gamma^{r_s}(x^{\ell}-\lambda) \in \langle f_0(x),f_1(x), \dots, f_s(x) \rangle$. Now we will proceed in a similar manner as above to obtain the required result.
    \end{enumerate}
\end{proof}

The following theorem provides a minimal spanning set for a $\lambda$-constacyclic code $\mathcal{C}$ of length $\ell$ over a FCR $\mathcal{R}$ along with its rank.
\begin{theorem}\label{thm:rank}
    Let $\mathcal{C}= \langle f_0(x),f_1(x), \dots , f_s(x) \rangle $ be a $\lambda$-constacyclic code of length $\ell$ over $\mathcal{R}$, where $f_0(x),f_1(x), \dots ,f_s(x)$ generators $\mathcal{C}$ as given in Theorem~\ref{thm:ccspanning}. For $0 \le i \le s$, let $S_i=\{ f_i(x), xf_i(x), \ldots, x^{n_{i+1}-n_{i}-1}f_i(x) \}, ~\text{where}~~n_{s+1}=\ell$. Then
    \begin{enumerate}
        \item ${S}=\bigcup_{i=0}^{s}S_i$ is a minimal spanning set of $\mathcal{C}$.
        \item $Rank(\mathcal{C})=\ell-n_0$, where $n_0=\deg(f_0(x))$ and $f_0(x)$ is a polynomial of minimal degree in $\mathcal{C}$.\label{part:rank}
    \end{enumerate}
\end{theorem}
\begin{proof}
    \begin{enumerate}
        \item Let $S'=\bigcup_{i=0}^s S_i'$, where $S_i'=\{f_i(x),xf_i(x), \dots , x^{\ell-n_i-1}f_i(x)\}~\text{for}~ 0 \le i \le s$. Clearly $S \subseteq S'$ and $S'$ spans $\mathcal{C}$. To prove that $S$ spans $\mathcal{C}$, it is sufficient to prove that $x^{n_{i+1}-n_{i}}f_i(x) \in Span(S)~~\text{for}~~0 \le i \le s-1$. We prove it by induction on $i$. For $i=0$, $x^{n_1-n_0}f_0(x)$ is a polynomial of degree $n_1$ in $\mathcal{C}$. Then, $x^{n_1-n_0}f_0(x)-\gamma^{r_0-r_1}u_1f_1(x)$ is a polynomial of degree less than $n_1$, where $u_1$ is a unit in $\mathcal{R}$. By Lemma~\ref{lem:first},  $x^{n_1-n_0}f_0(x)-\gamma^{r_0-r_1}u_1f_1(x)\in \langle f_0(x) \rangle$ which implies that $x^{n_1-n_0}f_0(x)-\gamma^{r_0-r_1}u_1f_1(x)= c_0(x)f_0(x)$, where $c_0(x) \in \mathcal{R}[x]$ with $\deg(c_0(x))<n_1-n_0$. Therefore, $x^{n_1-n_0}f_0(x)-\gamma^{r_0-r_1}u_1f_1(x) \in Span(S)$ which implies that $x^{n_1-n_0}f_0(x) \in Span(S)$. Hence, the result is true for $i=0$. Let us suppose that the result holds for all $i \le k-1$, i.e., $x^{n_{i+1}-n_{i}}f_i(x) \in Span(S)~~\text{for}~~0 \le i \le k-1$. Now we will prove the result for $i=k$. It is easy to see that $\deg(x^{n_{k+1}-n_{k}}f_k(x))=n_{k+1}~~\text{in}~~\mathcal{C}$. Then $x^{n_{k+1}-n_{k}}f_k(x)-\gamma^{r_k-r_{k+1}}u_{k+1}f_{k+1}(x)$ is a polynomial of degree less than $n_{k+1}$. Therefore, by Lemma~\ref{lem:first}, $x^{n_{k+1}-n_{k}}f_k(x)-\gamma^{r_k-r_{k+1}}u_{k+1}f_{k+1}(x)\in \langle f_0(x),f_1(x), \dots, f_k(x)\rangle$.  This implies that $x^{n_{k+1}-n_{k}}f_k(x)=\gamma^{r_k-r_{k+1}}u_{k+1}f_{k+1}(x)+c_0(x)f_0(x)+c_1(x)f_1(x)+\dots+c_k(x)f_k(x)$ where, $c_j(x)\in \mathcal{R}[x]$ such that $\deg(c_j(x))<n_{j+1}-n_j$ for all $0\le j \le k$. It follows that $c_j(x)f_j(x) \in Span(S)$ for all $0\le j \le k$. Therefore, $x^{n_{k+1}-n_{k}}f_k(x) \in Span(S)$. Thus, by mathematical induction $x^{n_{i+1}-n_{i}}f_i(x) \in Span(S)$ for all $0\le i \le s-1$. Hence, $S$ spans $\mathcal{C}$. The minimality of the spanning set $S$ can be easily proved by considering the degrees of various polynomials involved.
        \item By definition, $Rank(\mathcal{C})= |S|$, $S$ as a minimal spanning set of $\mathcal{C}$. Therefore, $|{S}|=\sum_{i=0}^{s}|S_i|=\sum_{i=0}^{s}(n_{i+1}-n_i)=n_{s+1}-n_0=\ell-n_0$. Hence $Rank(\mathcal{C})=\ell-n_0$.
        \end{enumerate}

\end{proof}
Below, we give a few examples to illustrate the above results.
\begin{example}\label{ex:first}
    Consider $\mathcal{R}=Z_{125}$. It possesses a maximal ideal generated by $5$, i.e.,$\langle\gamma\rangle=\langle5\rangle$ and has nilpotency index $\rho=3$. Then $\mathcal{C}=\langle25(x-2),5(x-2)^3\rangle$ forms a $2$-constacyclic code of length $5$ over $Z_{125}$, where
    \[
    \begin{aligned}
        f_0(x)&=25(x-2), \quad h_0(x)=(x-2), \quad r_0=2,\quad n_0=1 \\
        f_1(x)&=5(x-2)^3, \quad h_1(x)=(x-2)^3, \quad r_1=1,\quad n_1=3
    \end{aligned}
    \]
    The following set forms a minimal spanning set of $\mathcal{C}$
    \[
    S=\{f_0(x),xf_0(x),f_1(x),xf_1(x)\}
    \]
    where $|S|=4$. Also, $Rank(\mathcal{C})=5-1=4$.
\end{example}

\begin{example}\label{ex:second}
    Consider $\mathcal{R}=Z_{343}$. It possesses a maximal ideal generated by $7$, i.e. $\langle\gamma\rangle=\langle7\rangle$ and has nilpotency index $\rho=3$. Then $\mathcal{C}=\langle49,7(x^6-x^3+4)\rangle$ froms a $5$-constacyclic code of length $12$ over $Z_{343}$, where
    \[
    \begin{aligned}
        f_0(x)&=49, \quad h_0(x)=1, \quad r_0=2,\quad n_0=0 \\
        f_1(x)&=7(x^6-x^3+4), \quad h_1(x)=(x^6-x^3+4), \quad r_1=1,\quad n_1=6
    \end{aligned}
    \]
    The following set forms a minimal spanning set of $\mathcal{C}$
    \[
    \begin{split}
    S=\{
    f_0(x), xf_0(x), x^2f_0(x), x^3f_0(x), x^4f_0(x), x^5f_0(x),\\
    f_1(x), xf_1(x), x^2f_1(x), x^3f_1(x), x^4f_1(x), x^5f_1(x)
    \}
    \end{split}
    \]
    where $|S|=12$. Also, $Rank(\mathcal{C})=12-0=12$.
\end{example}

\begin{example}\label{ex:third}
    Consider $\mathcal{R}=Z_{289}$. It possesses a maximal ideal generated by $17$, i.e. $\langle\gamma\rangle=\langle17\rangle$ and has nilpotency index $\rho=2$. Then $\mathcal{C}=\langle17,(x^6+5x^5+8x^4+6x^3-4x^2-3x+2)\rangle$ forms a $10$-constacyclic code of length $7$ over $Z_{289}$, where
    \[
    \begin{aligned}
        f_0(x)&=17, \quad h_0(x)=1, \quad r_0=1,\quad n_0=0 \\
        f_1(x)&=x^6+5x^5+8x^4+6x^3-4x^2-3x+2,\\ h_1(x)&=x^6+5x^5+8x^4+6x^3-4x^2-3x+2,\quad r_1=0,\quad n_1=6
    \end{aligned}
    \]
    The following set forms a minimal spanning set of $\mathcal{C}$
    \[
    S=\{f_0(x),xf_0(x),x^2f_0(x),x^3f_0(x),x^4f_0(x),x^5f_0(x),f_1(x)\}
    \]
    where $|S|=7$. Also, $Rank(\mathcal{C})=7-0=7$.
\end{example}

\subsection{MHDR and MDS constacyclic codes}\label{sec:MHDR}
In this subsection, we obtain necessary and sufficient conditions under which a $\lambda$-constacyclic code is an MHDR as well as an MDS.

\medskip
Let $\mathcal{C}$ be a $\lambda$-constacyclic code of length $\ell$ over $\mathcal{R}$. For each integer $i$ with $0 \le i \le \rho-1$, $\mathrm{Tor}_i(\mathcal{C})=\{\,\overline{c(x)}\in \mathbb{F}_q[x]/\langle x^{\ell}- \overline{\lambda} \rangle \mid \gamma^{\,i}c(x) \in \mathcal{C} \,\}$ gives the $i$-th torsion code of $\mathcal{C}$. It can be easily observed that $\mathrm{Tor}_i(\mathcal{C})$ forms a $\overline{\lambda}$-constacyclic code of length $\ell$ over the residue field $\mathbb{F}_q$. Moreover, the degree associated with the generator polynomial of $\mathrm{Tor}_i(\mathcal{C})$ is referred to as the $i^{th}$ torsional degree of $\mathcal{C}$.

\begin{lemma}\label{lem:torsion}
    Let $\mathcal{C}= \langle f_0(x),f_1(x), \dots , f_s(x) \rangle $ be a $\lambda$-constacyclic code of length $\ell$ over $\mathcal{R}$, where $f_i(x),~0 \le i \le s$ are generators of $\mathcal{C}$ defined as $f_i(x)=\gamma^{r_i}h_i(x)$. Then  $0 \le i \le s,~\mathrm{Tor}_{r_i}(\mathcal{C})$ is a $\overline{\lambda}$-constacyclic code of length $\ell$ over the residue field $\mathbb{F}_q$ generated by $\overline{h_i(x)}$. Moreover, $n_i$ is the $r_i^{th}$ torsional degree of $\mathcal{C}$ and $dim(\mathrm{Tor}_{r_i}(\mathcal{C}))=\ell-n_i$.
\end{lemma}
\begin{proof}
    It is easy to observe that $\langle \overline{h_i(x)} \rangle \subseteq \mathrm{Tor}_{r_i}(\mathcal{C})$. Let $\overline{g(x)} \in \mathrm{Tor}_{r_i}(\mathcal{C})$ be a non-zero polynomial. Then $\gamma^{r_i}g(x) \in \mathcal{C}$. If $\deg(\gamma^{r_i}g(x))<n_i$, then by Lemma~\ref{lem:first}, $\gamma^{r_i}g(x) \in \langle f_0(x),f_1(x),\dots,f_{i-1}(x) \rangle$ from which it is observed that $\gamma^{r_i}g(x)= \sum_{k=0}^{i-1}c_k(x)\gamma^{r_k}h_k(x)$, where $c_k(x)$ belongs to the ring $\mathcal{R}[x]/\langle x^{\ell}-\lambda \rangle$ for $0\le k \le i-1$. It implies that $\gamma^{r_i}g(x)= \gamma^{r_i}\sum_{k=0}^{i-1}c_k(x)\gamma^{r_k-r_i}h_k(x)$. Since $r_k>r_i$ for every $0\le k \le i-1$, we obtain $\overline{g(x)}=0 ~\text{in}~ \mathbb{F}_q[x]/\langle x^{\ell}- \overline{\lambda} \rangle$, which is a contradiction. It follows that $\deg(\gamma^{r_i}g(x)) \ge n_i$.
    
    Now $h_i(x)$ be a monic polynomial with $\deg(h_i(x))=n_i$, which consequently implies that $\overline{h_i(x)}$ is a monic polynomial in $\mathrm{Tor}_{r_i}(\mathcal{C})$ with $\deg(\overline{h_i(x)})=n_i$. By division algorithm, there exist polynomials $\overline{q(x)} ~\text{and}~ \overline{r(x)}$ in the ring $\mathbb{F}_q[x]/\langle x^{\ell}- \overline{\lambda} \rangle$ such that $\overline{g(x)}=\overline{q(x)}~\overline{h_i(x)}+\overline{r(x)}$, where $\deg(\overline{r(x)}) < n_i~\text{or}~\overline{r(x)}=0$. Since $\overline{r(x)} \in \mathrm{Tor}_{r_i}(\mathcal{C})$, we must have $\overline{r(x)}=0$. Therefore, $\overline{g(x)}=\overline{q(x)}\overline{h_i(x)}$ which shows that $\overline{g(x)}\in \langle \overline{h_i(x)} \rangle$. Which results in $\mathrm{Tor}_{r_i}(\mathcal{C}) \subseteq \langle \overline{h_i(x)} \rangle$. Hence, $\mathrm{Tor}_{r_i}(\mathcal{C})= \langle \overline{h_i(x)} \rangle$.

    It is easy see that $\mathrm{Tor}_{r_i}(\mathcal{C})$ forms a $\overline{\lambda}$-constacyclic code of length $\ell$ over $\mathbb{F}_q$. Further, $dim(\mathrm{Tor}_{r_i}(\mathcal{C}))=\ell-\deg(\overline{h_i(x)})=\ell-n_i$.
\end{proof}

\begin{lemma}\label{lem:distance}
    Consider a $\lambda$-constacyclic code $\mathcal{C}$ of length $\ell$ over $\mathcal{R}$. Then
    \begin{enumerate}
        \item $d(\mathcal{C})=d(\mathrm{Tor}_{r_0}(\mathcal{C}))$.
        \item $Rank(\mathcal{C})=dim(\mathrm{Tor}_{r_0}(\mathcal{C}))$.
    \end{enumerate}
\end{lemma}
\begin{proof}
    \begin{enumerate}
        \item We know from Theorem~\ref{lem:torsion} that $\mathrm{Tor}_{r_0}(\mathcal{C})$ forms a $\overline{\lambda}$-constacyclic code of length $\ell$ over $\mathbb{F}_q$ whenever $\mathcal{C}$ is a $\lambda$-constacyclic code of length $\ell$ over $\mathcal{R}$. Let $\overline{f(x)} \in \mathrm{Tor}_{r_0}(\mathcal{C})$ be such that $wt(\mathrm{Tor}_{r_0}(\mathcal{C}))=wt(\overline{f(x)})$. Clearly, $wt(\overline{f(x)})=wt(\gamma^{r_0}f(x))$, where $\gamma^{r_0}f(x) \in \mathcal{C}$. Therefore, $wt(\mathrm{Tor}_{r_0}(\mathcal{C}))=wt(\overline{f(x)})=wt(\gamma^{r_0}f(x))\ge wt(\mathcal{C})$. Conversely, let $g(x)=g_0(x)+\gamma g_1(x)+ \dots +\gamma^{\rho-1}g_{\rho-1}(x) \in \mathcal{C}$ be such that $wt(\mathcal{C})=wt(g(x))$. Now $\gamma^{r_0}g(x) \in \mathcal{C}~\text{implies that}~ \overline{g_0(x)} \in \mathrm{Tor}_{r_0}(\mathcal{C})$. Therefore, $wt(\mathcal{C})=wt(g(x)) \ge wt(\overline{g_0(x)}) \ge wt(\mathrm{Tor}_{r_0}(\mathcal{C}))$. It follows that $wt(\mathcal{C})=wt(\mathrm{Tor}_{r_0}(\mathcal{C}))$ and therefore, $d(\mathcal{C})=d(\mathrm{Tor}_{r_0}(\mathcal{C}))$ because both $\mathcal{C}$ and $\mathrm{Tor}_{r_0}(\mathcal{C})$ are linear codes.
        \item From Theorem~\ref{thm:rank} and Lemma~\ref{lem:torsion}, it is clear that $Rank(\mathcal{C})=\ell-n_0=dim(\mathrm{Tor}_{r_0}(\mathcal{C}))$. Hence, $Rank(\mathcal{C})=dim(\mathrm{Tor}_{r_0}(\mathcal{C}))$.
    \end{enumerate}
\end{proof}

\begin{lemma}\label{lem:cardtor}\cite{mehrdad2012torsion}
    Consider a linear code $\mathcal{C}$ over FCR $\mathcal{R}$. Then $|\mathcal{C}|=\prod_{i=0}^{\rho-1}|\mathrm{Tor}_{i}(\mathcal{C})|$.
\end{lemma}
\begin{theorem}\label{thm:cardinality}
    Consider a $\lambda$-constacyclic code $\mathcal{C}= \langle f_0(x),f_1(x), \dots , f_s(x) \rangle $ of length $\ell$ over $\mathcal{R}$, where $f_0(x),f_1(x), \dots ,f_s(x)$ are generators of $\mathcal{C}$ as given in theorem~\ref{thm:ccspanning}. Then $|\mathcal{C}|=|\mathbb{F}_q|^{(\ell\rho-(\ell r_s+n_0(\rho-r_0)+\sum_{i=1}^sn_i(r_{i-1}-r_i)))}$, where $n_i,~ 0 \le i \le s$ are the torsional degrees of $\mathrm{Tor}_{r_i}(\mathcal{C})$.
\end{theorem}
\begin{proof}
    We know that $|\mathrm{Tor}_{i}(\mathcal{C})|=|\mathbb{F}_q|^{\ell-n_i}$, where $n_i$ is the degree associated with the generator polynomial of  $\mathrm{Tor}_{i}(\mathcal{C})$. Clearly $\mathrm{Tor}_{0}(\mathcal{C})=\mathrm{Tor}_{1}(\mathcal{C})= \dots =\mathrm{Tor}_{r_s-1}(\mathcal{C})={0}$. Therefore, $|\mathrm{Tor}_{0}(\mathcal{C})|=|\mathrm{Tor}_{1}(\mathcal{C})|= \dots =|\mathrm{Tor}_{r_s-1}(\mathcal{C})|=1$. Again $\mathrm{Tor}_{r_i}(\mathcal{C})=\mathrm{Tor}_{r_i+1}(\mathcal{C})= \dots =\mathrm{Tor}_{r_{i-1}-1}(\mathcal{C})=\langle \overline{h_i(x)} \rangle$ for $i=1,2,\dots,s$. Therefore, $|\mathrm{Tor}_{r_i}(\mathcal{C})|=|\mathrm{Tor}_{r_i+1}(\mathcal{C})|= \dots =|\mathrm{Tor}_{r_{i-1}-1}(\mathcal{C})|=|\mathbb{F}_q|^{\ell-n_i}$ for $i=1,2,\dots,s$. Also $\mathrm{Tor}_{r_0}(\mathcal{C})=\mathrm{Tor}_{r_0+1}(\mathcal{C})= \dots =\mathrm{Tor}_{\rho-1}(\mathcal{C})=\langle \overline{h_0(x)} \rangle$. Therefore, $|\mathrm{Tor}_{r_0}(\mathcal{C})|=|\mathrm{Tor}_{r_0+1}(\mathcal{C})|= \dots =|\mathrm{Tor}_{\rho-1}(\mathcal{C})|=|\mathbb{F}_q|^{\ell-n_0}$. Therefore, by Lemma~\ref{lem:cardtor}, $|\mathcal{C}|=\prod_{i=0}^{\rho-1}|\mathrm{Tor}_{i}(\mathcal{C})|=|\mathbb{F}_q|^{(\ell\rho-(\ell r_s+n_0(\rho-r_0)+\sum_{i=1}^sn_i(r_{i-1}-r_i)))}$.
\end{proof}

\begin{theorem}\label{thm:mhdr}
    Let $\mathcal{C}$ be a length $\ell$ $\lambda$-constacyclic code over a FCR $\mathcal{R}$. Then $\mathcal{C}$ is MHDR if and only if $\mathrm{Tor}_{r_0}(\mathcal{C})$ is MDS over $\mathbb{F}_q$.
\end{theorem}
\begin{proof}
    Suppose $\mathcal{C}$ is an MHDR code. Then $d(\mathcal{C}) = \ell - \mathrm{Rank}(\mathcal{C}) + 1 = n_0 + 1$. By Lemma~\ref{lem:distance}, $d(\mathrm{Tor}_{r_0}(\mathcal{C})) = d(\mathcal{C}) = n_0 + 1$. By Lemma~\ref{lem:torsion}, $\dim(\mathrm{Tor}_{r_0}(\mathcal{C})) = \ell - n_0$. As $\mathrm{Tor}_{r_0}(\mathcal{C})$ forms a $\overline{\lambda}$-constacyclic code over $\mathbb{F}_q$ having $\dim(\mathrm{Tor}_{r_0}(\mathcal{C})) = \ell - n_0$, $|\mathrm{Tor}_{r_0}(\mathcal{C})| = |\mathbb{F}_q|^{\ell - n_0} = |\mathbb{F}_q|^{\ell - d(\mathrm{Tor}_{r_0}(\mathcal{C})) + 1}$. Hence, $\mathrm{Tor}_{r_0}(\mathcal{C})$ is MDS over $\mathbb{F}_q$.

    Conversely, suppose $\mathrm{Tor}_{r_0}(\mathcal{C})$ is MDS code over $\mathbb{F}_q$. Then $|\mathrm{Tor}_{r_0}(\mathcal{C})|=|\mathbb{F}_q|^{(\ell-d(\mathrm{Tor}_{r_0}(\mathcal{C}))+1)}$ which shows that $dim(\mathrm{Tor}_{r_0}(\mathcal{C}))=\ell-d(\mathrm{Tor}_{r_0}(\mathcal{C}))+1$. From Lemma~\ref{lem:distance}, we have  $d(\mathcal{C})=d(\mathrm{Tor}_{r_0}(\mathcal{C}))$ and $Rank(\mathcal{C})=dim(\mathrm{Tor}_{r_0}(\mathcal{C}))$. Therefore, $d(\mathcal{C})=\ell-Rank(\mathcal{C})+1$. Hence, $\mathcal{C}$ is an MHDR code over $\mathcal{R}$.
\end{proof}

\begin{theorem}\label{thm:principal}
    Let $\mathcal{C}$ be a $\lambda$-constacyclic code of length $\ell$ over $\mathcal{R}$. Then the following are equivalent.
    \begin{enumerate}
        \item[(i)] $\mathcal{C}$ is MDS.
        \item[(ii)] $\mathrm{Tor}_{r_0}(\mathcal{C})$ is MDS and $\mathcal{C}$ is generated principally by a monic polynomial.
    \end{enumerate}
\end{theorem}

\begin{proof}
    Let $\mathcal{C}$ be an MDS code. Then $|\mathcal{C}|=|\mathcal{R}|^{\ell-d(\mathcal{C})+1}$. By Theorem~\ref{thm:cardinality} and Proposition~\ref{prop:first}, $|\mathbb{F}_q|^{(\ell\rho-(\ell r_s+n_0(\rho-r_0)+\sum_{i=1}^sn_i(r_{i-1}-r_i)))}=|\mathbb{F}_q|^{\rho(\ell-d(\mathcal{C})+1)}$. This implies that ${\ell\rho-(\ell r_s+n_0(\rho-r_0)+\sum_{i=1}^sn_i(r_{i-1}-r_i))}=\rho(\ell-d(\mathcal{C})+1)$. Therefore, we have
    \begin{equation}\label{eq:first}
         \ell r_s+n_0(\rho-r_0)+\sum_{i=1}^sn_i(r_{i-1}-r_i)=\rho(d(\mathcal{C})-1)
    \end{equation}
    Now, $n_0(\rho-r_s) \le n_0(\rho-r_0)+\sum_{i=1}^sn_i(r_{i-1}-r_i)$ as $n_0 < n_i~\text{for}~1 \le i \le s$. Therefore, $\ell r_s+ n_0(\rho-r_s)\le \ell r_s+n_0(\rho-r_0)+\sum_{i=1}^sn_i(r_{i-1}-r_i)$. Also, $d(\mathcal{C})-1 \le n_0$. Thus using equation~\ref{eq:first}, we have that
    \begin{align*}
        \ell r_s+  n_0(\rho-r_s) \le \ell r_s+n_0(\rho-r_0)+\sum_{i=1}^sn_i(r_{i-1}-r_i)=\rho(d(\mathcal{C})-1) \le \rho n_0.
        \end{align*}
       This implies that $\ell r_s+  n_0(\rho-r_s) \le \rho n_0$,
      which further implies that $r_s(\ell-n_0)\le 0$.
    
    It follows that $r_s=0$ as $r_s \ge 0~\text{and}~\ell-n_0>0$ .
    Therefore, equation~\ref{eq:first} simplifies to $n_0(\rho-r_0)+\sum_{i=1}^sn_i(r_{i-1}-r_i)=\rho(d(\mathcal{C})-1)$, which can further be written as
    \begin{equation*}
        n_0(\rho-r_0)+\sum_{i=1}^sn_i(r_{i-1}-r_i)=((\rho-r_0)+\sum_{i=1}^s(r_{i-1}-r_i))(d(\mathcal{C})-1).
    \end{equation*}
    It implies that
    \begin{equation*}
         (\rho-r_0)+\sum_{i=1}^s(n_i-(d(\mathcal{C})-1))(r_{i-1}-r_i)=0.
    \end{equation*}
    
    Since $\rho-r_0>0~\text{and}~r_{i-1}-r_i>0~\text{for all}~1 \le i \le s$, we obtain $n_i=d(\mathcal{C})-1$ for all $0\le i \le s$. This together with Lemma~\ref{lem:distance} implies that $n_0=d(\mathcal{C})-1=d(\mathrm{Tor}_{r_0}(\mathcal{C}))-1$. Therefore, $|\mathrm{Tor}_{r_0}(\mathcal{C})|=|\mathbb{F}_q|^{\ell-n_0}=|\mathbb{F}_q|^{\ell-d(\mathrm{Tor}_{r_0}(\mathcal{C}))+1}$. Hence, $\mathrm{Tor}_{r_0}(\mathcal{C})$ is MDS. Now $n_i=d(\mathcal{C})-1$ for all $0\le i \le s$ and $n_0<n_1<\dots <n_s$ can hold simultaneously only when $s$ is zero. Therefore, $r_0=r_s=0$. Consequently, $\mathcal{C}=\langle h_0(x) \rangle$, i.e., $\mathcal{C}$ is generated principally by a monic polynomial. Hence, $(i) \implies(ii)$.

    Conversely, suppose that $\mathrm{Tor}_{r_0}(\mathcal{C})$ is MDS and $\mathcal{C}$ is generated principally by a monic polynomial. Then, by Theorem~\ref{thm:cardinality}, $|\mathcal{C}|=|\mathbb{F}_q|^{\rho(\ell-n_0)}$. Since $\mathrm{Tor}_{r_0}(\mathcal{C})$ is MDS, $n_0=d(\mathrm{Tor}_{r_0}(\mathcal{C}))-1=d(\mathcal{C})-1$. 
    It implies that $|\mathcal{C}|=|\mathbb{F}_q|^{\rho(\ell-(d(\mathcal{C})-1))}=|\mathcal{R}|^{\ell-d(\mathcal{C})+1}$, which follows that $\mathcal{C}$ is MDS. Hence $(ii) \implies(i)$.
    \end{proof}

    The following Corollary is a direct implication of Theorems~\ref{thm:mhdr} and~\ref{thm:principal}.

\begin{corollary}
     Consider a $\lambda$-constacyclic code $\mathcal{C}$ of length $\ell$ over $\mathcal{R}$. Then $\mathcal{C}$ is MDS implies that $\mathcal{C}$ is MHDR.
\end{corollary}
The converse of the above Corollary need not be true. For example consider a $2$-constacyclic code $\mathcal{C}=\langle25(x-2),5(x-2)^3\rangle$ of length $5$ over $Z_{125}$. Clearly, $\mathcal{C}$ is MHDR by Theorem~\ref{thm:mhdr} but not MDS by Theorem~\ref{thm:principal}.

\medskip
Necessary and sufficient conditions for a $\lambda$-constacyclic code of length $p^m$ to be MDS over $\mathbb{F}_q$ have been given in Theorem 3.2 by Dinh et al.~\cite{dinh2020mds}. The following theorem extends this result to a $\lambda$-constacyclic code of arbitrary length by using similar arguments as given by Dinh et al.~\cite{dinh2020mds}.
\begin{theorem}\label{thm:dinhdistance}
    Consider a $\lambda$-constacyclic code $\mathcal{C}_j$ of arbitrary length $\ell=np^m$ with $(n,p)=1$ over $\mathbb{F}_q$ such that it is generated by a polynomial of degree $j$. Then $\mathcal{C}_j$ is an MDS constacyclic code if and only if \item[(i)] $j \in \{ 0, 1, \dots, p-1\} ~\text{when}~m=1$. \item[(ii)] $j \in \{0,1,p^m-1\}~\text{when}~m \ge 2$.
\end{theorem}
\begin{proof}
    The proof proceeds in a manner similar to that of Theorem 3.2 of~\cite{dinh2020mds}.
\end{proof}

The following Corollary follows directly from Theorems~\ref{thm:mhdr} and~\ref{thm:dinhdistance}.
\begin{corollary}\label{cor:nnote}
    Let $\mathcal{C}= \langle f_0(x),f_1(x), \dots , f_s(x) \rangle $ be a $\lambda$-constacyclic code of length $\ell=np^m$ with $(n,p)=1$ over $\mathcal{R}$, where $f_0(x),f_1(x), \dots ,f_s(x)$ are generators of $\mathcal{C}$ as given in theorem~\ref{thm:ccspanning}. Then $\mathcal{C}$ is MHDR if and only if \item[(i)] $n_0 \in \{ 0, 1, \dots, p-1\} ~\text{when}~m=1$. \item[(ii)] $n_0 \in \{0,1,p^m-1\}~\text{when}~m \ge 2$.
\end{corollary}

Some examples are provided below to support our results.
\begin{example}\label{ex:fourth}
    Suppose $\mathcal{R}=F_5+\gamma F_5$ having nilpotency index $\rho=2$. Consider a $2$-constacyclic code $\mathcal{C}=\langle \gamma (x-2)^2,(x-2)^4 \rangle$ of length $5$ over $\mathcal{R}$. By Theorem~\ref{thm:rank} $Rank(\mathcal{C})=3$ and $d(\mathcal{C})=3$. Also, $\mathrm{Tor}_{r_0}(\mathcal{C})= \langle (x-2)^2 \rangle$ forms a $2$-constacyclic code of length $5$ over $F_5$ with $dim(\mathrm{Tor}_{r_0}(\mathcal{C}))=3$ and $d(\mathrm{Tor}_{r_0}(\mathcal{C}))=3$. Clearly $|\mathrm{Tor}_{r_0}(\mathcal{C})|=5^3$ and $|\mathbb{F}_q|^{(\ell-d(\mathrm{Tor}_{r_0}(\mathcal{C}))+1)}=5^3$. Therefore, $\mathrm{Tor}_{r_0}(\mathcal{C})$ is MDS code. Hence, by Theorem~\ref{thm:mhdr}, $\mathcal{C}$ is MHDR over $\mathcal{R}$.
\end{example}

\begin{example}\label{ex:fifth}
   Over $Z_9$ having nilpotency index $\rho=2$, Consider an $8$-constacyclic code $\mathcal{C}=\langle 3(x^2+1)\rangle$ of length $4$. By Theorem~\ref{thm:rank}, $\mathrm{Tor}_{r_0}(\mathcal{C})=\langle x^2+1\rangle$ forms an $8$-constacyclic code of length $4$ over $\mathbb{F}_3$ with $\dim(\mathrm{Tor}_{r_0}(\mathcal{C}))=2$ and $d(\mathrm{Tor}_{r_0}(\mathcal{C}))=2$. Clearly $|\mathrm{Tor}_{r_0}(\mathcal{C})|=3^2$ and $|\mathbb{F}_3|^{(\ell-d(\mathrm{Tor}_{r_0}(\mathcal{C}))+1)}=3^3$. By Corollary~\ref{cor:nnote}, $\mathrm{Tor}_{r_0}(\mathcal{C})$ fails to be an MDS code. Hence, by Theorem~\ref{thm:mhdr}, $\mathcal{C}$ cannot be an MHDR code over $\mathbb{Z}_9$.

\end{example}
\begin{example}
     Consider $\mathcal{R}=Z_{289}$. The maximal ideal of $\mathcal{R}$ is $\langle\gamma\rangle=\langle17\rangle$ and nilpotency index $\rho=2$. Then, consider a $10$-constacyclic code $\mathcal{C}=\langle x^6+5x^5+8x^4+6x^3-4x^2-3x+2\rangle$ of length $7$ over $Z_{289}$. Clearly $\mathrm{Tor}_{r_0}(\mathcal{C})=\langle x^6+5x^5+8x^4+6x^3-4x^2-3x+2\rangle$. By Corollary~\ref{cor:nnote}, $\mathrm{Tor}_{r_0}(\mathcal{C})$ is MDS. Also, $\mathcal{C}$ is generated principally by a monic polynomial. By Theorem~\ref{thm:principal}, $\mathcal{C}$ is also MDS.
\end{example}
It can be seen by Theorem~\ref{thm:mhdr} that the constacyclic codes given in Examples~\ref{ex:first},~\ref{ex:second} and~\ref{ex:third} of subsection~\ref{subsec:one} are all MHDR.

\section{Conclusion}\label{sec:conc}
In the above study, the generators of the constacyclic code $\mathcal{C}$ of arbitrary length $\ell$ over a FCR $\mathcal{R}$ were explicitly determined. Using these generators, a minimal spanning set was obtained along with the rank of the code $\mathcal{C}$. Moreover, we also established necessary and sufficient conditions under which a constacyclic code of arbitrary length becomes MHDR and MDS. Some examples are also provided to support our results.

\printbibliography

@article{dalal2024mds,
  title={MDS and MHDR cyclic codes over finite chain rings},
  author={Dalal, Monika and Dutt, Sucheta and Sehmi, Ranjeet},
  journal={Journal of Mathematics},
  volume={2024},
  number={1},
  pages={4540992},
  year={2024},
  publisher={Wiley Online Library}
}

@article{monika2021cyclic,
  title={On cyclic codes over finite chain rings},
  author={Monika and Dutt, Sucheta and Sehmi, Ranjeet},
  journal={Journal of Physics: Conference Series},
  volume={1850},
  number={1},
  pages={012010},
  year={2021},
  organization={IOP Publishing}
}

@article{dinh2017repeated,
  title={Repeated-root constacyclic codes of prime power lengths over finite chain rings},
  author={Dinh, Hai Q and Nguyen, Hien DT and Sriboonchitta, Songsak and Vo, Thang M},
  journal={Finite Fields Appl.,},
  volume={43},
  pages={22--41},
  year={2017},
  publisher={Elsevier}
}

@article{castagnoli1991repeated,
  title={On repeated-root cyclic codes},
  author={Castagnoli, Guy and Massey, James L and Schoeller, Philipp A and Von Seemann, Niklaus},
  journal={IEEE Trans. Inform. Theory},
  volume={37},
  number={2},
  pages={337--342},
  year={1991},
  publisher={IEEE}
}

@article{van1991repeated,
  title={Repeated-root cyclic codes},
  author={van Lint, Jacobus H},
  journal={IEEE Trans. Inform. Theory},
  volume={37},
  number={2},
  pages={343--345},
  year={1991},
  publisher={IEEE}
}

@article{calderbank1994z4,
  title={The $Z_4$-linearity of Kerdock, Preparata, Goethals and related codes},
  author={Calderbank, AR and Hammons Jr, AR and Kumar, P Vijay and Sloane, NJA and Sol{\'e}, P},
  journal={IEEE Trans. Inform. Theory},
  volume={40},
  number={2},
  pages={301--319},
  year={1994}
}

@article{norton2000structure,
  title={On the structure of linear and cyclic codes over a finite chain ring},
  author={Norton, Graham H and S{\u{a}}l{\u{a}}gean, Ana},
  journal={Appl. Algebra Engrg. Comm. Comput.,},
  volume={10},
  number={6},
  pages={489--506},
  year={2000},
  publisher={Springer}
}

@article{norton2002hamming,
  title={On the Hamming distance of linear codes over a finite chain ring},
  author={Norton, Graham H and Salagean, Ana},
  journal={IEEE Trans. Inform. Theory},
  volume={46},
  number={3},
  pages={1060--1067},
  year={2002},
  publisher={IEEE}
}

@article{dinh2004cyclic,
  title={Cyclic and negacyclic codes over finite chain rings},
  author={Dinh, Hai Quang and L{\'o}pez-Permouth, Sergio R},
  journal={IEEE Trans. Inform. Theory},
  volume={50},
  number={8},
  pages={1728--1744},
  year={2004},
  publisher={IEEE}
}

@article{dinh2005negacyclic,
  title={Negacyclic codes of length $2^s$ over galois rings},
  author={Dinh, Hai Q},
  journal={IEEE Trans. Inform. Theory},
  volume={51},
  number={12},
  pages={4252--4262},
  year={2005},
  publisher={IEEE}
}

@article{dinh2009constacyclic,
  title={Constacyclic Codes of Length $2^{s} $ Over Galois Extension Rings of $F_2 + uF_2$},
  author={Dinh, Hai Q},
  journal={IEEE Trans. Inform. Theory},
  volume={55},
  number={4},
  pages={1730--1740},
  year={2009},
  publisher={IEEE}
}

@article{sualuagean2006repeated,
  title={Repeated-root cyclic and negacyclic codes over a finite chain ring},
  author={S{\u{a}}l{\u{a}}gean, Ana},
  journal={Discrete Appl. Math.,},
  volume={154},
  number={2},
  pages={413--419},
  year={2006},
  publisher={Elsevier}
}

@article{dinh2008linear,
  title={On the linear ordering of some classes of negacyclic and cyclic codes and their distance distributions},
  author={Dinh, Hai Q},
  journal={Finite Fields Appl.,},
  volume={14},
  number={1},
  pages={22--40},
  year={2008},
  publisher={Elsevier}
}

@article{dinh2010constacyclic,
  title={Constacyclic codes of length $p^s$ over $F_{p^m}+ uF_{p^m}$},
  author={Dinh, Hai Q},
  journal={Journal of Algebra},
  volume={324},
  number={5},
  pages={940--950},
  year={2010},
  publisher={Elsevier}
}

@article{cao2018constacyclic,
  title={Constacyclic codes of length $np^s$ over $F_{p^m}+ uF_{p^m}$},
  author={Cao, Yonglin and Cao, Yuan and Dinh, Hai Q and Fu, Fang-Wei and Gao, Jian and Sriboonchitta, Songsak},
  journal={Adv. Math. Commun.,},
  volume={12},
  number={2},
  pages={231--262},
  year={2018}
}

@article{raka2015class,
  title={A class of constacyclic codes over a finite field-II},
  author={Raka, Madhu},
  journal={Indian J. Pure Appl. Math.,},
  volume={46},
  number={6},
  pages={809--825},
  year={2015},
  publisher={Springer}
}

@article{chen2012constacyclic,
  title={Constacyclic codes over finite fields},
  author={Chen, Bocong and Fan, Yun and Lin, Liren and Liu, Hongwei},
  journal={Finite Fields Appl.,},
  volume={18},
  number={6},
  pages={1217--1231},
  year={2012},
  publisher={Elsevier}
}

@article{dinh2012repeated,
  title={Repeated-root constacyclic codes of length $2p^s$},
  author={Dinh, Hai Q},
  journal={Finite Fields Appl.,},
  volume={18},
  number={1},
  pages={133--143},
  year={2012},
  publisher={Elsevier}
}

@article{cao2013constacyclic,
  title={On constacyclic codes over finite chain rings},
  author={Cao, Yonglin},
  journal={Finite Fields Appl.,},
  volume={24},
  pages={124--135},
  year={2013},
  publisher={Elsevier}
}

@article{dinh2013structure,
  title={Structure of repeated-root constacyclic codes of length $3p^s$ and their duals},
  author={Dinh, Hai Q},
  journal={Discrete Mathematics},
  volume={313},
  number={9},
  pages={983--991},
  year={2013},
  publisher={Elsevier}
}

@article{chen2014repeated,
  title={Repeated-root constacyclic codes of length $lp^s$ and their duals},
  author={Chen, Bocong and Dinh, Hai Q and Liu, Hongwei},
  journal={Discrete Appl. Math.,},
  volume={177},
  pages={60--70},
  year={2014},
  publisher={Elsevier}
}

@article{sharma2018structure,
  title={On the structure and distances of repeated-root constacyclic codes of prime power lengths over finite commutative chain rings},
  author={Sharma, Anuradha and Sidana, Tania},
  journal={IEEE Trans. Inform. Theory},
  volume={65},
  number={2},
  pages={1072--1084},
  year={2018},
  publisher={IEEE}
}

@article{chen2016constacyclic,
  title={Constacyclic codes of length $2p^s$ over $F_{p^m}+ uF_{p^m}$},
  author={Chen, Bocong and Dinh, Hai Q and Liu, Hongwei and Wang, Liqi},
  journal={Finite Fields Appl.,},
  volume={37},
  pages={108--130},
  year={2016},
  publisher={Elsevier}
}

@book{mcdonald1974finite,
  title={Finite rings with identity},
  author={McDonald, Bernard R},
  publisher= {Marcel Dekker},
  location={New York},
  year={1974}
}

@inproceedings{mehrdad2012torsion,
  author    = {Mehrdad, M},
  title     = {Torsion Codes Over a Finite Chain Rings},
  booktitle = {Second Workshop on Algebra and its Applications},
  year      = {2012}
}

@article{dinh2020mds,
  title={MDS constacyclic codes of prime power lengths over finite fields and construction of quantum MDS codes},
  author={Dinh, Hai Q and ElDin, Ramy Taki and Nguyen, Bac T and Tansuchat, Roengchai},
  journal={Internat. J. Theoret. Phys.,},
  volume={59},
  number={10},
  pages={3043--3078},
  year={2020},
  publisher={Springer}
}

\end{document}